\documentclass[prl, amsmath, amssymb, preprintnumbers, twocolumn, superscriptaddress]{revtex4}

\usepackage{amssymb, amsthm}
\usepackage{color}
\usepackage{graphicx}
\usepackage{subfigure}
\usepackage{amssymb}
\usepackage{ulem}

\newcommand{\Tr}[1]{\mathrm{Tr} #1}

\newtheorem*{theorem}{Lemma}

\newcommand{\mytitle}{Local quantum thermal susceptibility}

\begin{document}

\title{\mytitle}

\author{Antonella De Pasquale}
\email{antonella.depasquale@sns.it}
\affiliation{NEST, Scuola Normale Superiore and Istituto Nanoscienze-CNR, 
  Piazza dei Cavalieri 7, I-56126 Pisa, Italy}

\author{Davide Rossini}
\affiliation{NEST, Scuola Normale Superiore and Istituto Nanoscienze-CNR, 
  Piazza dei Cavalieri 7, I-56126 Pisa, Italy}

\author{Rosario Fazio}
\affiliation{ICTP, Strada Costiera 11, 34151 Trieste, Italy}
\affiliation{NEST, Scuola Normale Superiore and Istituto Nanoscienze-CNR, 
  Piazza dei Cavalieri 7, I-56126 Pisa, Italy}

\author{Vittorio Giovannetti}
\affiliation{NEST, Scuola Normale Superiore and Istituto Nanoscienze-CNR, 
  Piazza dei Cavalieri 7, I-56126 Pisa, Italy}

\begin{abstract}
  Thermodynamics relies on the possibility to describe systems composed of a large number 
  of constituents in terms of few macroscopic variables. 
  Its foundations are rooted into the paradigm of statistical mechanics, 
  where thermal properties originate from averaging procedures which smoothen out local details. 
  While undoubtedly successful, elegant and formally correct, this approach carries over 
  an operational problem: what is the precision at which such variables are inferred, 
  when technical/practical limitations restrict our capabilities to local probing? 
  Here we introduce the local quantum thermal susceptibility, a quantifier for the best 
  achievable accuracy for temperature estimation via local measurements. 
  Our method relies on basic concepts of quantum estimation theory, providing an operative 
  strategy to address the local thermal response of arbitrary quantum systems at equilibrium. 
  At low temperatures it highlights the local distinguishability of the ground state 
  from the excited sub-manifolds, thus providing a method to locate quantum phase transitions.
\end{abstract}


\maketitle

The measurement of temperature is a key aspect in science, technology, and in our daily life. 
Many ingenious solutions have been designed to approach different situations and 
required accuracies~\cite{Pekola}. What is the ultimate limit to the precision at which the temperature 
of a macroscopic state can be determined? An elegant answer to this question is offered 
by estimation theory~\cite{cramer, paris_book, paris1}: The precision is related 
to the heat capacity of the system~\cite{zanardi_heatcapacity1, zanardi_heatcapacity2}. 

In view of the groundbreaking potentialities offered by present-day 
nanotechnologies~\cite{gao, nano1, nano2, nano3, thermo_exp, thermo_exp1} 
and the need to control the temperature at the nano-scale, it is highly relevant to question 
whether the heat capacity is still the relevant (fundamental) precision limit to small-scale thermometry. 
The extensivity of the heat capacity is a consequence of the growing volume-to-surface ratio 
with the size~\cite{Huang}. However, at a microscopic level such construction may present 
some limitations~\cite{Hill, Hill1}. Moreover a series of theoretical efforts recently concentrated 
on a self-consistent generalization of the classical thermodynamics to small-scale physics, 
where quantum effects become 
predominant~\cite{gemmer, Allahverdyan, Allahverdyan1, Nieuwenhuizen, Lutz, LeHur, Horodecki}. 
In particular, a lot of attention has been devoted to the search for novel methods of precision 
nanothermometry that could exploit the essence of quantum 
correlations~\cite{Qest_JC, Qest_JC2, marzolino, thermo_theory, mandarino, thermo_theory_XYmodel}. 
In this context, the possibility to correctly define the thermodynamical limit, and therefore 
the existence of the temperature in the quantum regime, has been thoroughly investigated. 
It has been shown that the minimal subset of an interacting
quantum system which can be described as a canonical ensemble, with the same temperature 
of the global system, depends not only on the strength of the correlations 
within the system, but also on the temperature itself~\cite{hartmann1, hartmann1b, hartmann2}.
Using a quantum information-oriented point of view, this phenomenon has been also 
highlighted in Gaussian fermionic and bosonic states, by exploiting quantum fidelity 
as the figure of merit~\cite{saez, saez2}.
Furthermore, the significant role played by quantum correlations has been recently discussed with 
specific attention to spin- and fermonic-lattice systems with short-range interactions~\cite{eisert}. 

In this paper we propose a quantum-metrology approach to thermometry,
through the analysis of the {\it local} sensitivity
of generic quantum systems to their {\it global} temperature. 
Our strategy does not assume any constraint neither on the structure 
of the local quantum state, nor on the presence of strong quantum fluctuations 
within the system itself, while it rather moves from the observation that the temperature 
is a parameter which can be addressed only via indirect measurements. 
Specifically we introduce a new quantity which we dub 
Local Quantum Thermal Susceptibility (LQTS), and which gauges how efficiently 
the thermal equilibrium of a system is perceived by its subparts. 
More precisely, given a quantum system ${\cal AB}$ in a thermal equilibrium state, 
the LQTS $\mathfrak{S}_{\cal A}$ is a response functional which quantifies the highest achievable 
accuracy for estimating the system temperature $T$ through local measurements 
performed on a selected subsystem ${\cal A}$ of ${\cal AB}$ (see Fig.~\ref{fig:thermo}).

\begin{figure}[!b]
  \includegraphics[height=0.45\columnwidth]{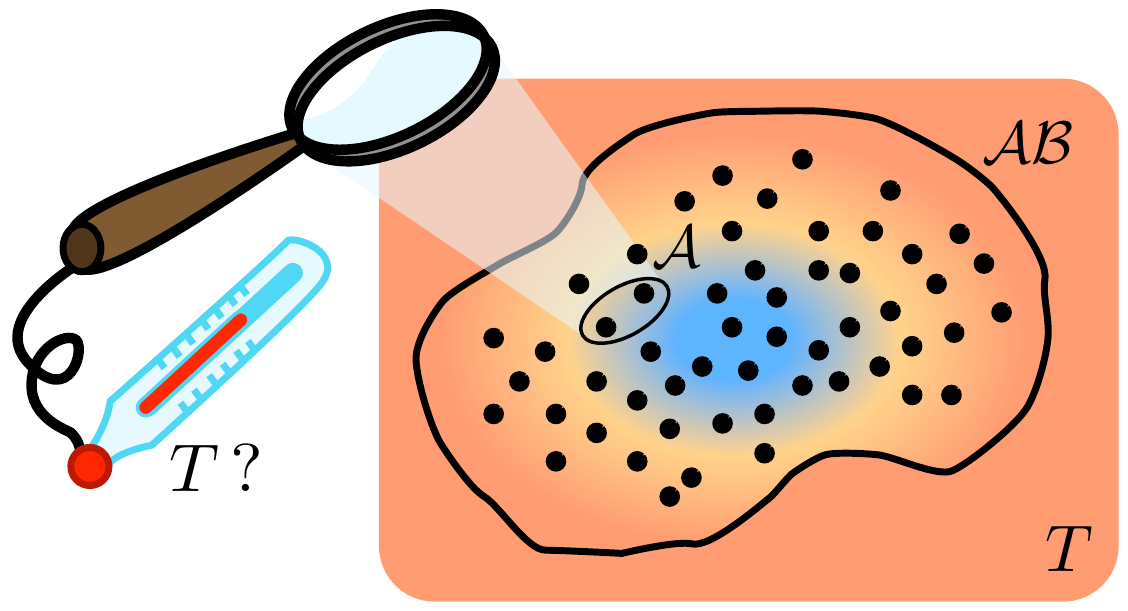}
  \caption{\textbf{The operationally grounded strategy associated to the LQTS functional.}
    A composite quantum system $\cal S$ is in thermal equilibrium with a bath at temperature $T$. 
    The LQTS functional~\eqref{LQTS_def} measures the highest achievable accuracy 
    in the estimation $T$ under the hypothesis to perform only local measurements 
    on the subsystem $\cal A$ of ${\cal AB}$.}
  \label{fig:thermo}
\end{figure}

The LQTS is in general not extensive with respect to the size of ${\cal A}$, 
yet it is an increasing function of the latter, and it reduces to the system heat capacity 
in the limit where the probed part coincides with the whole system ${\cal AB}$.
In the low-temperature limit, we shall also see that the LQTS is sensitive to the local 
distinguishability between the ground state and the first excited subspace of the composite 
system Hamiltonian. In this regime, even for a tiny size of the probed subsystem, 
our functional is able to predict the behaviour of the heat capacity and in particular 
to reveal the presence of critical regions. This naturally suggests the interpretation 
of $\mathfrak{S}_{\cal A}$ as a sort of mesoscopic version of the heat capacity which replaces 
the latter in those regimes where extensivity breaks down.

\section{Results}

\paragraph{ \textbf{The functional.}} 

Let us consider a bipartite quantum system 
${\cal AB}$ at thermal equilibrium, composed of two subsystems ${\cal A}$ and ${\cal B}$,
and described by the canonical Gibbs ensemble $\rho_{\beta}= e^{-\beta H}/{\cal Z}_\beta$.
Here $H=H_{\cal A} + H_{\cal B} + H_{\cal AB}^{\rm (int)}$ is the system Hamiltonian, 
which in the general case will include both local ({\it i.e.}, $H_{\cal A}$ and $H_{\cal B}$) 
and interaction ({\it i.e.}, $H_{\cal AB}^{\rm (int)}$) terms, 
while ${\cal Z}_\beta = \mathrm{Tr}[e^{-\beta H}] = \sum_i e^{-\beta E_i}$ 
denotes the associated partition function ($\beta=1/k_B T$ is the inverse temperature 
of the system, $k_B$ the Boltzmann constant, and $\{ E_i\}$ the eigenvalues of $H$). 
In this scenario we are interested in characterizing how the actual temperature $T$ 
is perceived locally on ${\cal A}$. 

For this purpose we resort to quantum metrology~\cite{giovannetti} and define 
the LQTS of subsystem ${\cal A}$ as
\begin{equation} \label{LQTS_def}
  \mathfrak{S}_{\cal A}[\rho_{\beta}]:= 8 \lim_{\varepsilon \to 0} \frac{1-\mathcal{F} 
    \big( \rho^{\cal A}_\beta, \rho^{\cal A}_{\beta+\varepsilon} \big) }{\varepsilon^2}\,,
\end{equation}
where $\mathcal{F}\left(\rho, \sigma\right) = \mathrm{Tr}[\sqrt{\sqrt{\rho} \, \sigma \sqrt{\rho}}]$ 
is the fidelity between two generic quantum states $\rho$ and $\sigma$~\cite{peres, jozsa}. 
The quantity~\eqref{LQTS_def} corresponds to the quantum Fisher 
information (QFI)~\cite{paris_book, paris1} 
for the estimation of $\beta$, computed on the reduced state 
$\rho^{\cal A}_\beta=\mathrm{\Tr}_{\mathcal{B}}[\rho_\beta]$.
It gauges how modifications on the global system temperature are affecting ${\cal A}$, 
the larger being $\mathfrak{S}_{\cal A}[\rho_{\beta}]$ the more sensitive being the subsystem response. 
More precisely, through the quantum Cram\'er-Rao inequality, $\mathfrak{S}_{\cal A}[\rho_{\beta}]$ 
quantifies the ultimate precision limit to estimate the temperature $T$, 
by means of any possible local (quantum) measurement on subsystem ${\cal A}$. 
In the specific, it defines an asymptotically achievable lower bound,
\begin{equation} 
  {\Delta T}^{\cal A} \geq k_B T^2/ \sqrt{N \mathfrak{S}_{\cal A}[\rho_{\beta}]} \, ,
\end{equation}
on the root-mean-square error 
${\Delta T}^{\cal A}=\sqrt{\mathbb{E}[(T^{\mathrm{est}}-T)^2]}$ of a generic local 
estimation strategy, where $T^{\mathrm{est}}$ is the estimated value of $T$, 
$\mathbb{E}[x]$ is the expectation value for a random variable $x$, 
and $N$ is the number of times the local measurement is repeated. 

By construction, $\mathfrak{S}_{\cal A}[\rho_{\beta}]$ is a positive quantity 
which diminishes as the size of ${\cal A}$ is reduced, the smaller being the portion 
of the system we have access to, the worse being the accuracy we can achieve. 
More precisely, given ${\cal A}'$ a proper subset of ${\cal A}$, we have 
$\mathfrak{S}_{{\cal A}'}[\rho_{\beta}] \leq \mathfrak{S}_{\cal A}[\rho_{\beta}]$.
In particular, when ${\cal A}$ coincides with the whole system ${\cal AB}$, 
equation~\eqref{LQTS_def} reaches its maximum value and becomes equal to the variance of the energy,
\begin{equation}
  \mathfrak{S}_{\cal AB}[\rho_{\beta}] = \mathrm{Tr} [\rho_\beta H^2] -\mathrm{Tr} [\rho_\beta H]^2 \,,
\label{eq:Heat}
\end{equation}
which depends only on the spectral properties of the system and which coincides with the system 
heat capacity~\cite{zanardi_heatcapacity1, zanardi_heatcapacity2} (note that, rigorously speaking, 
the LQTS quantifies the sensitivity of the system to its inverse temperature $\beta$; 
the corresponding susceptibility to $T = 1/(k_B \beta)$ gets a $k_B^{-2} T^{-4}$ correction term, 
which also enters the standard definition of the heat capacity). 

An explicit evaluation of the limit in equation~\eqref{LQTS_def} can be obtained 
via the Uhlmann's theorem~\cite{uhlmann} (see Methods for details). 
A convenient way to express the final result can be obtained by introducing an ancillary system 
${\cal A'B'}$ isomorphic to ${\cal AB}$ and the purification of $\rho_\beta$ defined as
\begin{equation}
  \label{PUREVEC}
  |\rho_\beta\rangle=\sum_i \frac{ e^{- \beta E_i/2}}{\sqrt{{\cal Z}_\beta}} 
  |E_i\rangle_{\cal A \cal B} \otimes |E_i\rangle_{\cal A' \cal B'}\;,
\end{equation}
where $H\!=\!\sum_i \!E_i |E_i\rangle_{\cal A \cal B} \langle E_i|$ is the spectral decomposition 
of the system Hamiltonian. 
It can then be proved that
\begin{equation}
  \label{eq:LQTSfinal}
  \mathfrak{S}_{\cal A}[\rho_{\beta}]=\mathfrak{S}_{\cal AB}[\rho_{\beta}] 
  - \sum_{j<k} \frac{(\lambda_j - \lambda_k)^2}{\lambda_j + \lambda_k} |\langle e_k|H| e_j \rangle |^2\,,
\end{equation} 
where $\{ |e_j\rangle\}$ are the eigenvectors of the reduced density matrix 
$\mathrm{Tr}_{\cal A'}[|\rho_\beta\rangle \langle\rho_\beta|]$ living on ${\cal A B B'}$, 
obtained by taking the partial trace of $|\rho_\beta\rangle$ with respect to the 
ancillary system ${\cal A'}$, while $\{ \lambda_j\}$ are the corresponding eigenvalues 
(which, by construction, coincide with the eigenvalues of~$\rho_\beta^{\cal A}$).

Equation~\eqref{eq:LQTSfinal} makes it explicit the ordering between 
$\mathfrak{S}_{\cal A}[\rho_{\beta}]$ and $\mathfrak{S}_{\cal AB}[\rho_{\beta}]$: 
the latter is always greater than the former due to the negativity of the second 
contribution appearing on the rhs. 
Furthermore, if $H$ does not include interaction terms ({\it i.e.}, $H_{\cal AB}^{\rm (int)}=0$), 
one can easily verify that $\mathfrak{S}_{\cal A}[\rho_{\beta}]$ reduces to the variance 
of the local Hamiltonian of ${\cal A}$, and is given by the heat capacity 
of the Gibbs state $e^{-\beta H_{\cal A}}/{\cal Z}^{\cal A}_\beta$ which, in this special case, 
represents $\rho_\beta^{\cal A}$, {\it i.e.}, 
$\mathfrak{S}_{\cal A}[\rho_{\beta}] = \mathrm{Tr} [\rho^{\cal A}_\beta H_{\cal A}^2] -\mathrm{Tr} [\rho^{\cal A}_\beta H_{\cal A}]^2= - \frac{\partial^2}{\partial \beta^2} \ln {\cal Z}^{\cal A}_\beta$.
Finally we observe that in the high-temperature regime ($\beta \to 0$) 
the expression~\eqref{eq:LQTSfinal} simplifies yielding 
\begin{equation}
  \mathfrak{S}_{\cal A}[\rho_{\beta}] \! = \! \frac{1}{d_{\cal A}} \mathrm{Tr} \!\! 
  \left [\tilde{H}_{\cal A}^2 - 2 \beta \tilde{H}_{\cal A} \frac{\mathrm{Tr}_{\cal B} [H^2] }{d_{\cal B}} 
    + \beta \tilde{H}_{\cal A}^3\right] \! 
  + \mathcal{O}(\beta^2)\,,
\end{equation}
where $d_{\cal A}$ and $d_{\cal B}$ denote the Hilbert space dimensions of 
${\cal A}$ and ${\cal B}$ respectively, and we defined 
$\tilde{H}_{\cal A} = \mathrm{Tr}_{\cal B}[H]/d_{\cal B}$ having set, without loss of generality,
$\mathrm{Tr}[H]=0$.

\begin{figure}[!t]
  \includegraphics[height=0.45\columnwidth]{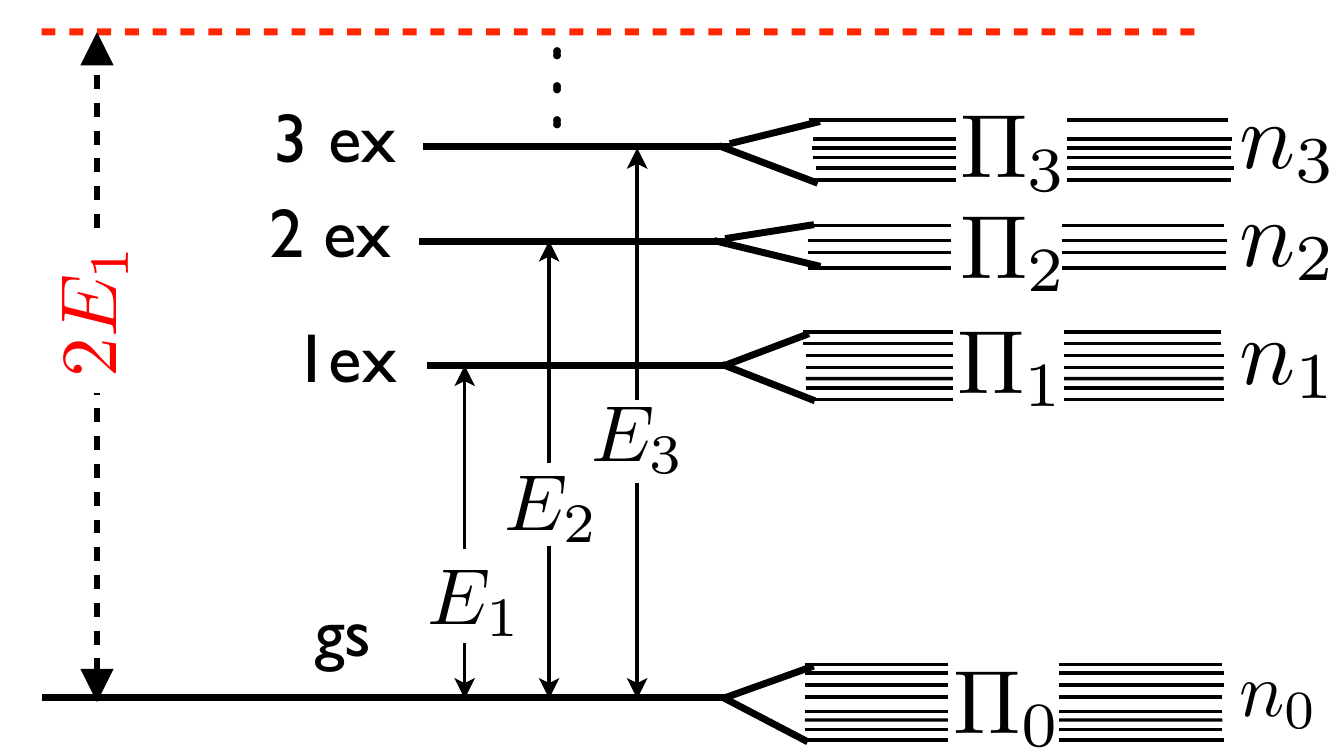}
  \caption{\textbf{Schematic representation of the low-energy spectrum for a generic 
      many-body quantum system.}
    For simplicity the ground-state (gs) energy $E_0$ is set to zero. 
    Here $\Pi_i$ denotes the normalized projector on the eigenspaces of energy $E_i$, 
    which can be $n_i$-fold degenerate.}
  \label{fig:spectrum}
\end{figure}

\paragraph{ \textbf{A measure of state distinguishability.}} 

In the low-temperature regime, the LQTS can be used to characterize how much the ground state 
of the system ${\cal AB}$ differs from the first excited subspaces when observing it 
locally on~${\cal A}$. This is a direct consequence of the fact that 
the QFI (which we used to define our functional) accounts for the degree of statistical 
distinguishability between two quantum states (in our case the reduced 
density matrices $\rho_{\beta}^{\cal A}$ and 
 $\rho_{\beta+\varepsilon}^{\cal A}$) differing by an infinitesimal change 
in the investigated parameter (in our case the inverse temperature $\beta$). 
Therefore for $\beta \to \infty$,
the LQTS can be thought as a quantifier of the local distinguishability 
among the lowest energy levels in which the system is frozen. 

To clarify this point, let us consider the general scenario 
depicted in Fig.~\ref{fig:spectrum}, where we only discuss the physics 
of the ground state (with energy $E_0 = 0$) and of
the lowest excited levels with energy $E_i$ bounded by twice 
the energy of the first excited level, $E_i \leq 2 E_1$.
The degeneracy of each considered energy eigenstate is denoted by $n_i$. 
From equation~\eqref{eq:LQTSfinal} 
it then follows that up to first order in the parameter $e^{- \beta E_1}$ we get
\begin{equation}
  \mathfrak{S}_{\cal A}[\rho_{\beta}] = 
  \sum_{i} \frac{n_i}{n_0} E_i^2 e^{- \beta E_i} \Big( 1 - \mathrm{Tr}[P_0^{\cal A} \Pi_i^{\cal A}] \Big) 
  + \mathcal{O} \left( e^{- 2\beta E_1} \right) . 
  \label{EQQ}
\end{equation} 
Here $\Pi_i^{\cal A} = \mathrm{Tr}_{\cal B}[\Pi_i]$, where $\Pi_i$ is the normalized projector 
on the degenerate subspace of energy $E_i$. 
Moreover, $P^{\cal A}_0$ is the span of the local subspace associated to the ground state, 
{\it i.e.}, $P_0^{\cal A} = \sum_{j=1}^{n_0^{\cal A}} |\phi_j \rangle_{\cal A} \langle \phi_j|$, 
with $\Pi_0^{\cal A} = \sum_{j=1}^{n_0^{\cal A}} p_j | \phi_{j} \rangle_{\cal A} \langle \phi_j|$,
and $n_0^{\cal A}$ being the number of non-zero eigenvalues ($p_j>0$) of $\Pi_0^{\cal A}$.

Equation~\eqref{EQQ} can be interpreted as follows. Our capability of measuring $\beta$ 
relies on the distinguishability between the states $\rho_\beta^{\cal A}$ 
and $\rho_{\beta+\varepsilon}^{\cal A}$, with $\varepsilon \ll \beta$. 
In the zero-temperature limit, the system lies in the ground state 
and locally reads as $\Pi_0^{\cal A}$, while at small temperatures 
the lowest energy levels start to get populated. 
If their reduced projectors $\Pi_i^{\cal A}$ ($i \geq1$) are not completely contained 
in the span of $\Pi_0^{\cal A}$, that is $\mathrm{Tr}[P_0^{\cal A} \, \Pi_i^{\cal A}]\neq 1$, 
there exist some local states whose population is null for $T=0$ and greater than zero 
at infinitesimal temperatures. 
Such difference implies that the first order in $\mathfrak{S}_{\cal A}$ 
does not vanish. On the contrary, if the reduced projectors $\Pi_i^{\cal A}$ 
are completely contained in the span of $\Pi_0^{\cal A}$, 
that is $\mathrm{Tr}[P_0^{\cal A} \, \Pi_i^{\cal A}]=1$, 
we can distinguish $\rho_\beta^{\cal A}$ from $\rho_{\beta+\varepsilon}^{\cal A}$ 
only thanks to infinitesimal corrections $\mathcal{O} (\! \exp(-2\beta E_1))$ 
to the finite-valued populations of the lowest energy levels 
(see Methods for an explicit evaluation of the latter).
In conclusion, the quantity $\mathfrak{S}_{\cal A}[\rho_{\beta \to \infty}]$ acts as 
a thermodynamical indicator of the degree of distinguishability 
between the ground-state eigenspace
and the lowest energy levels in the system Hamiltonian.

\paragraph{\textbf{LQTS and phase estimation.}}

A rather stimulating way to interpret equation~\eqref{eq:LQTSfinal} comes from the observation that, 
in the extended scenario where we have purified ${\cal AB}$ as in equation~\eqref{PUREVEC}, 
the global variance~\eqref{eq:Heat} formally coincides with the QFI 
$\mathfrak{F}_{\cal ABA'B'}(|\rho_\beta^{(\varphi)}\rangle)$ associated with the estimation 
of a phase $\varphi$ which, for given $\beta$, has been imprinted into the system 
${\cal ABA'B'}$ by a unitary transformation $e^{-i H' \varphi/2}$, 
with $H'$ being the analogous of $H$ on the ancillary system ${\cal A'B'}$, 
{\it i.e.}, $\mathfrak{S}_{\cal AB}[\rho_{\beta}]=\mathfrak{F}_{\cal ABA'B'}(|\rho_\beta^{(\varphi)}\rangle)$ 
where $|\rho_\beta^{(\varphi)}\rangle=e^{-i H' \varphi/2}|\rho_\beta\rangle$~\cite{paris1,giovannetti,caves}. 
Interestingly enough a similar connection can be also established with the second term 
appearing in the rhs of equation~\eqref{eq:LQTSfinal}: indeed the latter coincides 
with the QFI $\mathfrak{F}_{\cal B A' B'}( |\rho_\beta^{(\varphi)}\rangle )$ which defines 
the Cram\'er-Rao bound for the estimation of the phase $\varphi$ 
of $|\rho_\beta^{(\varphi)}\rangle$, under the constraint of having access only on 
the subsystem ${\cal B A' B'}$ ({\it i.e.}, that part of the global system 
which is complementary to ${\cal A}$). 
Accordingly we can thus express the LQTS as the difference between these two QFI 
phase estimation terms, the global one vs.~the local one or, by a simple rearrangement 
of the various contributions, construct the following identity 
\begin{equation}
  \mathfrak{S}_{\cal A}[\rho_{\beta}] + \mathfrak{F}_{\cal B A' B'} ( |\rho_\beta^{(\varphi)}\rangle )
  = \langle \Delta H^2 \rangle_\beta \;,
\end{equation}
which establishes a complementarity relation between the temperature estimation 
on ${\cal A}$ and the phase estimation on its complementary counterpart 
${\cal B A' B'}$, by forcing their corresponding accuracies to sum up to 
the energy variance $\langle \Delta H^2 \rangle_\beta$ of the global system~\eqref{eq:Heat}.

\paragraph{ \textbf{Local thermometry in many-body systems.}} 

We have tested the behaviour of our functional on two models of quantum spin chains, 
with a low-energy physics characterized by the emergence of quantum phase transitions (QPTs) 
belonging to various universality classes~\cite{sachdev}.

\begin{figure*}[!t]
  \includegraphics[width=1\textwidth]{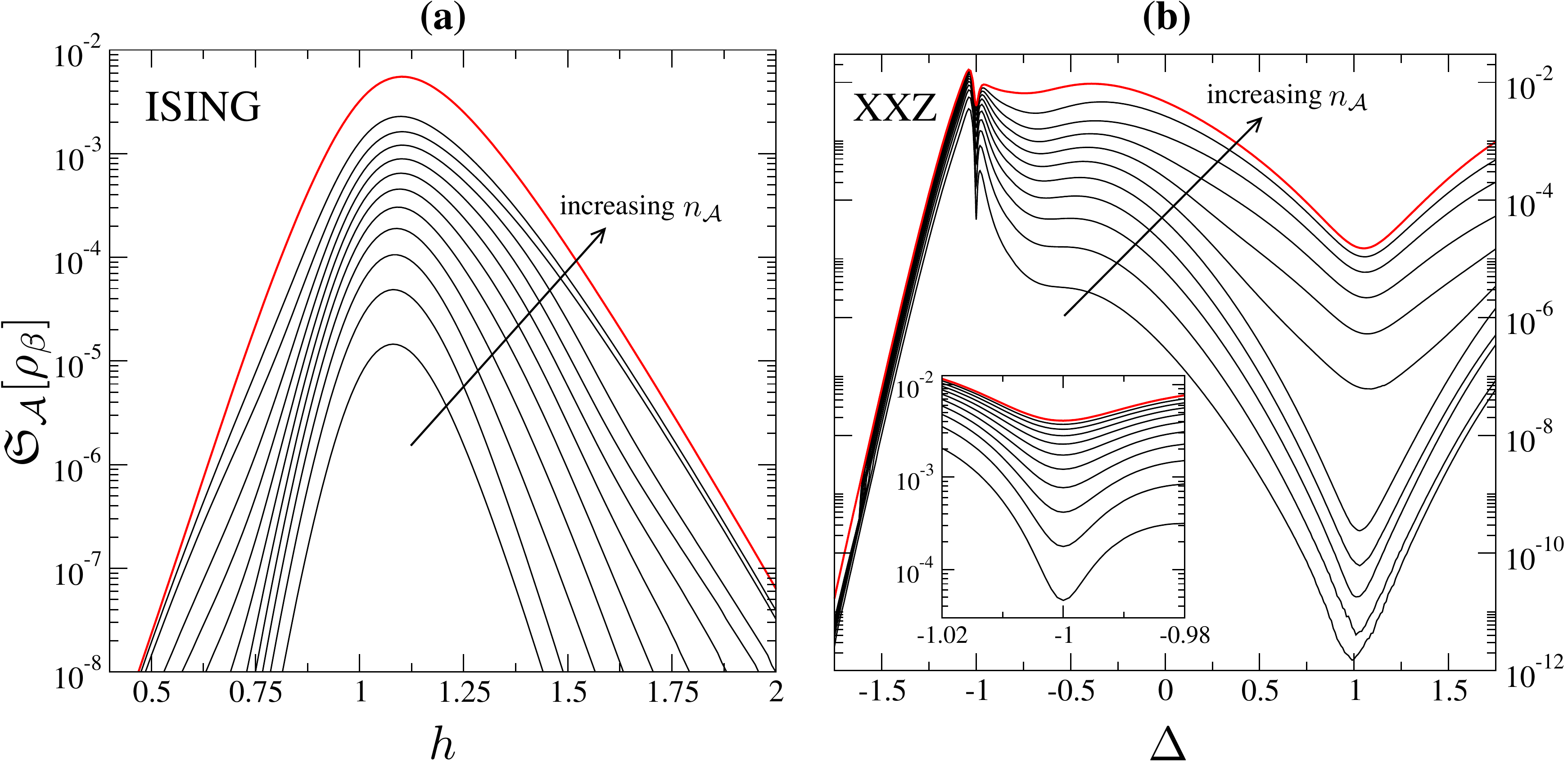} 
   \caption{\textbf{The LQTS in the Ising and the Heisenberg models at low temperature.} 
    We numerically computed the LQTS of equation~\eqref{LQTS_def} in the low-temperature limit 
    for a chain with $L=12$ sites in the following two cases:
    ({\bf a}) the Ising model as a function of the transverse field $h$;
    ({\bf b}) the Heisenberg XXZ chain as a function of the anisotropy $\Delta$. 
    The uppermost (red) curve corresponds to the global quantum thermal susceptibility,
    that is the heat capacity.
    The other curves stand for different sizes $n_{\cal A}$ of the measured subsystem 
    ${\cal A}$ of ${\cal AB}$ ($n_{\cal A}$ increases along the direction of the arrow). 
    The inset in panel ({\bf b}) magnifies the data around $\Delta = -1$.
    In the XXZ model, the LQTS with $n_{\cal A} = 1$ can be proved to rigorously vanish.
    The inverse temperature has been fixed in both cases at $\beta = 9$.}
  \label{fig:LQTS_Ising}
\end{figure*}

Specifically, we consider the quantum spin-$1/2$ Ising and Heisenberg chains, 
in a transverse magnetic field $h$ and with a $z$-axis anisotropy $\Delta$ respectively, 
\begin{eqnarray}
  H_{\rm Ising} & \! = \! & 
  - J \sum_i \left[ \sigma_i^{x} \sigma_{i+1}^{x} + h \sigma_i^{z} \right] , \label{eq:Ising}\\
  H_{\rm XXZ} & \! = \! & J \sum_i \left[ (\sigma_i^{x} \sigma_{i+1}^{x} 
    + \sigma_i^{y} \sigma_{i+1}^{y}) 
    \! + \! \Delta \, \sigma_i^{z} \sigma_{i+1}^{z} \right] \! . \label{eq:XXZ}
\end{eqnarray}
Here $\sigma^{\alpha}_i$ denotes the usual Pauli matrices ($\alpha = x,y,z$) on the $i$th site, 
and periodic boundary conditions have been assumed.
We set $J=1$ as the system's energy scale.
At zero temperature, the model~\eqref{eq:Ising} presents a $\mathbb{Z}_2$-symmetry breaking
phase transition at $\vert h_c \vert=1$ belonging to the Ising universality class.
The Hamiltonian~\eqref{eq:XXZ} has a critical behaviour for $-1 \leq \Delta \leq 1$,
while it presents a ferromagnetic or antiferromagnetic ordering elsewhere. 
In the latter case, the system exhibits a first-order QPT in correspondence to the ferromagnetic 
point $\Delta_{\rm f} = -1$, and a continuous QPT of the Kosterlitz-Thouless type
at the antiferromagnetic point $\Delta_{\rm af} = 1$.

Figure~\ref{fig:LQTS_Ising} displays the small-temperature limit of $\mathfrak{S}_{\cal A}[\rho_{\beta}]$ 
for the two models above, numerically computed by exploiting expression~\eqref{eq:LQTS_numerical} in Methods. 
We first observe that, as expected, for all the values of $h$ and $\Delta$ the LQTS monotonically 
increases with increasing the number $n_{\cal A}$ of contiguous spins belonging to 
the measured subsystem ${\cal A}$. More interestingly, we find that even when $\cal A$ 
reduces to two or three sites, its thermal behaviour qualitatively reproduces the same features 
of the global system (represented in both models by the uppermost curve).
In particular, even at finite temperatures and for systems composed of twelve sites, 
the LQTS is sensitive to the presence of critical regions where quantum fluctuations 
overwhelm thermal ones. 
The reminiscence of QPTs at finite temperatures has been already discussed via a 
quantum-metrology approach, through the analysis of the Bures metric tensor in 
the parameter space associated to the temperature and the external parameters~\cite{zanardiT}. 
The diagonal element of such tensor referring to infinitesimal variations 
in temperature, corresponds to the thermal susceptibility of the whole system. 
The latter quantity has been recently studied for the XY model~\cite{thermo_theory_XYmodel}, 
showing its sensitivity to critical points of Ising universality class. 

In the low-temperature regime such global sensitivity can be understood within 
the Landau-Zener (LZ) formalism~\cite{LZ}.
This consists of a two-level system, whose energy gap $\Delta E$
varies with respect to an external control parameter $\Gamma$, and presents 
a minimum $\Delta E_{\rm min}$ in correspondence to some specific value $\Gamma_c$.
Conversely, the global heat capacity~\eqref{eq:Heat} may exhibit a maximum 
or a local minimum at $\Gamma_c$,
according to whether $\Delta E_{\rm min}$ is greater or lower than the value of $\Delta E^\star$ 
at which the expression $\mathfrak{S}_{\cal AB}[\rho_\beta] \equiv \langle \Delta H^2 \rangle_\beta$ 
is maximum in $\Delta E$, respectively.
Indeed it can be shown that $\langle \Delta H^2 \rangle_\beta$ for a two-level system 
exhibits a non-monotonic behaviour as a function of $\Delta E$, at fixed $\beta$ (see Methods). 
Quite recently, an analogous mechanisms has been also pointed out for the global
heat capacity in the Lipkin-Meshkov-Glick model~\cite{mandarino}. 
The LZ formalism represents a simplified picture of the mechanism underlying QPTs in many-body systems. 
However, by definition the temperature triggers the level statistics and the equilibrium properties 
of physical systems. Therefore, both the heat capacity of the global 
system~\cite{zanardi_heatcapacity1, zanardi_heatcapacity2} 
and the LQTS of its subsystems are expected to be extremely sensitive to the presence of critical regions 
in the Hamiltonian parameter space. 

We also performed a finite-size scaling analysis of $\mathfrak{S}_{\cal A}[\rho_{\beta}]$
as a function of the size of the measured subsystem.
For slightly interacting systems, one expects the LQTS to be well approximated 
by the heat capacity of ${\cal A}$ (at least when this subsystem is large enough). 
The latter quantity should scale linearly with its size $n_{\cal A}$.
This is indeed the case for the Ising model~\eqref{eq:Ising}, where a direct calculation
of $\langle \Delta H^2 \rangle_\beta$ can be easily performed~\cite{thermo_theory_XYmodel}.
Our data for the scaling of the stationary points of $\mathfrak{S}_{\cal A}[\rho_{\beta}]$
close to QPTs suggest that significant deviations from a linear growth with $n_{\cal A}$ 
are present (see the Supplementary Information). This indicates that correlations 
cannot be neglected for the sizes and the systems considered here.

\section{Discussion}

We have proposed a theoretical approach to temperature locality based on quantum estimation theory. 
Our method deals with the construction of the \textit{local quantum thermal susceptibility}, 
which operationally highlights the degree at which the thermal equilibrium of the global system 
is perceived locally, avoiding any additional hypothesis on the local structure of the system.
This functional corresponds to the highest achievable accuracy up to which it is possible 
to recover the system temperature at thermal equilibrium via local measurements. 
Let us remark that, even if in principle the Cram\'er-Rao bound is achievable, 
from a practical perspective it represents a quite demanding scenario, as it requires 
the precise knowledge of the Hamiltonian, the possibility to identify and perform the 
optimal measurements on its subsystems, and eventually a large number of copies of the system.
However, in this manuscript we have adopted a more theoretical perspective, and focused on 
the geometrical structure of the quantum statistical model underlying local thermalization.

In the low-temperature regime our functional admits an interpretation as a measure 
of the local state distinguishability between the spaces spanned by the Hamiltonian 
ground state and its first energy levels. 
Furthermore we established a complementarity relation between the highest achievable accuracy 
in the local estimation of temperature and of a global phase, by showing that 
the corresponding accuracies associated with complementarity subsystems 
sum up to heat capacity of the global system. 
Finally, we considered two prototypical many-body systems featuring quantum phase transitions, 
and studied their thermal response at low temperatures. 
On one hand, we found that optimal measurements on local systems provide reliable predictions 
on the global heat capacity. 
On the other hand, our functional is sensitive to the presence of critical regions, 
even though the total system may reduce to a dozen of components 
and the measured subsystem to one or two sites. Let us remark that most of the results 
presented herewith do not refer to any specific choice of the interaction Hamiltonian,
$H_{\cal AB}^{\rm (int)}$ between $\cal A$ and $\cal B$. As an interesting implementation of our scheme, 
we foresee the case of non-thermalizing interactions~\cite{stace, zwierz}, whose potentialities 
for precision thermometry have been already unveiled.

We conclude by noticing that, while in this article we focused on temperature, the presented approach 
can be extended to other thermodynamic quantities (like entropy, pressure, chemical potential etc.).
Furthermore it seems plausible to adopt quantum estimation strategies to tackle the problem 
of providing self-consistent definition of heat and work for microscopic systems (see Ref.~\cite{campisi} 
and references therein). One of the main difficulties to this end derives indeed from the fact that 
these thermodynamic functionals are processes rather than state variables:
accordingly there are no quantum observables associated with them. 
In this respect a quantum-estimation strategy not explicitly referred to a specific 
quantum observable, but rather bearing the geometrical traits of the Hilbert space associated 
to the explored systems, appears as a valid chance to close this gap.

\section{Methods} 

\paragraph{ \textbf{Derivation of equation~(\ref{eq:LQTSfinal}).}} 
Let us recall the definition of the LQTS for a given subsystem $\cal A$ of
a global system $\cal S$ at thermal equilibrium:
\begin{equation} 
  \label{LQTS_def2}
  \mathfrak{S}_{\cal A}[\rho_{\beta}]= 8 \lim_{\varepsilon \to 0} \frac{1-\mathcal{F} 
    \big( \rho^{\cal A}_\beta, \rho^{\cal A}_{\beta+\varepsilon} \big) }{\varepsilon^2}\,,
\end{equation}
where $\mathcal{F}\left(\rho, \sigma\right) = \mathrm{Tr}[\sqrt{\sqrt{\rho} \, \sigma \sqrt{\rho}}]$ 
is the fidelity between two generic quantum states $\rho$ and $\sigma$.
According to the Uhlmann's theorem~\cite{uhlmann}, we can compute $\cal F$ as 
\begin{equation}
  \label{eq:fidelity}
  \mathcal{F}\left(\rho^{\cal A}_\beta, \rho^{\cal A}_{\beta+\varepsilon}\right) 
  = \max_{|\rho^{\cal A}_\beta\rangle,|\rho^{\cal A}_{\beta+\varepsilon}\rangle} |\langle \rho^{\cal A}_\beta 
  |\rho^{\cal A}_{\beta+\varepsilon}\rangle|,
\end{equation} 
where the maximization involves all the possible purifications 
$|\rho^{\cal A}_\beta\rangle$ and $|\rho^{\cal A}_{\beta+\varepsilon}\rangle$ 
of $\rho^{\cal A}_\beta$ and $\rho^{\cal A}_{\beta+\varepsilon}$ 
respectively through an ancillary system $a$. A convenient choice 
is to set $a = {\cal B A' B'}$, with ${\cal A'B'}$ isomorphic to ${\cal AB}$. 
We then observe that, by construction, the vector $|\rho_\beta\rangle$ of equation~\eqref{PUREVEC}, 
besides being a purification of $\rho_\beta$, is also a particular purification 
of $\rho^{\cal A}_\beta$. We can now express the most generic purification of the latter as 
\begin{equation}
  |\rho^{\cal A}_\beta\rangle = ( \openone_{\cal{A}} \otimes V ) |\rho_\beta\rangle 
  = ( \openone_{\cal{A}} \otimes V ) \frac{ e^{- \beta H/2}}{\sqrt{{\cal Z}_\beta}} |\openone_E \rangle\,,
\end{equation}
where $V$ belongs to the set of unitary transformations on~$a$, where $\openone_{\cal X}$ 
represents the identity operator on the system ${\cal X}$, and where in the last equality 
we introduced the vector 
$|\openone_E \rangle = \sum_i |E_i\rangle_{\cal A \cal B}|E_i\rangle_{\cal A' \cal B'}$, 
$|E_i\rangle_{\cal AB}$ being the eigenvectors of $H$. 
We can thus write the fidelity~\eqref{eq:fidelity} as
\begin{equation}
  \mathcal{F} = \sqrt{\frac{\cal Z_\beta}{\cal Z_{\beta+\varepsilon}}}\max_{V} \Big|\langle \rho_\beta| 
  \big[ ( \openone_{\cal{A}} \otimes V ) \big( e^{ - \varepsilon H/2 }\otimes \openone_{\cal A' \cal B'} \big) \big] |\rho_{\beta}\rangle \Big| \,.
\end{equation}
Since we are interested in the small-$\varepsilon$ limit, without loss of generality we set 
$V=\exp(i \, \varepsilon \, \Omega)$, with $\Omega$ being an Hermitian operator on the ancillary 
system $a$. It comes out that, up to corrections of order $\mathcal{O}(\varepsilon^3)$, the LQTS reads
\begin{eqnarray}
   \mathfrak{S}_{\cal A}[\rho_{\beta}] = \mathfrak{S}_{\cal AB}[\rho_{\beta}] 
   + 4 \min_\Omega \Big \{ \mathrm{Tr} [\rho_\beta^a \Omega^2]-\mathrm{Tr} [\rho_\beta^a \Omega]^2 \nonumber \\+ \frac{i}{2} \langle \rho_\beta| [ \Omega, H] | \rho_\beta \rangle\Big\}\,,
  \label{eq:LQTS_min}
\end{eqnarray} 
where we have defined $\rho_\beta^a = \mathrm{Tr}_{\cal A}[|\rho_\beta\rangle \langle\rho_\beta|]$. 
By differentiating the trace with respect to $\Omega$ we determine the minimization condition 
for it, yielding
\begin{equation}
  \label{eq:Omega}
  (\Omega - \omega) \rho_\beta^a + \rho_\beta^a (\Omega - \omega) = Q\,,
\end{equation}
with $\omega=\mathrm{Tr}[\rho_\beta^a \Omega]$ and $Q= - \frac{i}{2} \big[ H', \rho_\beta^a \big]$,
$H'$ being the analogous of $H$ which acts on ${\cal A'B'}$ 
(by construction $H |\rho_\beta\rangle = H' |\rho_\beta\rangle$).
Equation~\eqref{eq:Omega} explicitly implies that $\Omega$ does not depend on $\omega$, 
which, without loss of generality, can be set to zero. 
Moreover, it enables to rewrite the LQTS in equation~\eqref{eq:LQTS_min} as
\begin{equation}
  \label{eq:LQTS1}
  \mathfrak{S}_{\cal A}[\rho_{\beta}] = \mathfrak{S}_{\cal AB}[\rho_{\beta}] -4 \, \mathrm{Tr}[ \rho_\beta^a \, \Omega^2]\,. 
\end{equation} 
The solution of the operatorial equation~\eqref{eq:Omega} can be found by applying 
the Lemma presented at the end of this section, yielding
\begin{eqnarray} 
\Omega &=& \Omega_0 + \frac{i}{2}\left(PH'R+RH'P\right) \nonumber \\
 & &+\frac{i}{2}\left[\sum_{n=1}^{+\infty} (-1)^n \left(P{\rho_\beta^a}^{\circleddash n} H' {\rho_\beta^a}^n - h.c.\right) \right]
\end{eqnarray}
with $\Omega_0$ being an operator which anti-commutes with $\Omega$, ${\rho_\beta^a}^{\circleddash m}$ 
being the Moore-Penrose pseudoinverse of $\rho_\beta^a$ to the power $m$, 
while $R$ being the projector on kernel of $\rho_\beta^a$, 
and $P =\openone_a-R$ being is complementary counterpart. 
By substituting this expression in equation~\eqref{eq:LQTS1}, we finally get
\begin{equation}
\mathfrak{S}_{\cal A}[\rho_{\beta}] = 2 \sum_{j, k}   \frac{\lambda_j \lambda_k}{\lambda_j + \lambda_k} |\langle e_k|H' | e_j \rangle |^2 - \Tr [\rho_\beta H]^2\,,
\end{equation}
where $\rho_\beta^a=\sum_{i}\lambda_i |e_i \rangle \langle e_i|$ is the spectral decomposition 
of $\rho_\beta^a$, sharing the same spectrum with $\rho_\beta^{\cal A}$. 
The expression above holds for both invertible and not invertible $\rho_\beta^a$. 
To the latter scenario belongs the case in which $H=H_{\cal A}+H_{\cal B}$, 
where one can easily prove that the LQTS reduces to the variance of the local 
Hamiltonian $H_{\cal A}$, {\it i.e.} 
$\mathfrak{S}_{\cal A}= \mathrm{Tr} [\rho_\beta H_{\cal A}^2] -\mathrm{Tr} [\rho_\beta H_{\cal A}]^2$ 
(notice that the non-zero eigenvalues of $\rho_\beta^a$ are 
$\lambda_i=e^{- \beta E_i^{\cal A}}/ {\cal Z}_\beta^{\cal A}$ which correspond to 
$| e_i \rangle = |E_i^{\cal A}\rangle \otimes |\rho_\beta^{\cal B}\rangle$, 
being $H_{\cal A}=\sum_{i}E_i^{\cal A} | E_i^{\cal A} \rangle \langle E_i^{\cal A}|$, 
${\cal Z}_\beta^{\cal A}=\Tr[e^{-\beta H_{\cal A}}]$ and $|\rho_\beta^{\cal B}\rangle$ 
the purification of $\rho_\beta^{\cal B}$ through the ancillary system $\cal{B'}$). 
The expression above can be also rewritten as
\begin{equation}
  \mathfrak{S}_{\cal A}[\rho_{\beta}]=\mathfrak{S}_{\cal AB}[\rho_{\beta}] - \sum_{j<k} \frac{(\lambda_j - \lambda_k)^2}{\lambda_j + \lambda_k} |\langle e_k|H' | e_j \rangle |^2\,,
\end{equation}
which can be cast into equation~\eqref{eq:LQTSfinal} by simply exploiting the fact that 
the system is symmetric with respect to the exchange of ${\cal AB}$ with ${\cal A'B'}$.

It is finally useful to observe that LQTS can be also expressed in terms of 
the eigenvectors of ${\cal A}$, $\rho_\beta^{\cal A}=\sum_{i} \lambda_i |g_i \rangle \langle g_i|$ as: 
\begin{equation}
  \mathfrak{S}_{\cal A}[\rho_{\beta}]= 2 \sum_{j, k} \frac{|\Tr[\rho_\beta H | g_j \rangle \langle g_k|]|^2}{\lambda_j + \lambda_k} - \Tr[\rho_\beta H]^2,
  \label{eq:LQTS_numerical}
\end{equation}
where we have used the Schmidt decomposition of $|\rho_\beta \rangle$, with respect 
to bipartition ${\cal A}a$, 
\begin{equation}
  |\rho_\beta \rangle = \sum_{i} \sqrt{\lambda_i}|g_i \rangle |e_i \rangle. 
\end{equation}
In particular, expression~\eqref{eq:LQTS_numerical} can be exploited in order to 
numerically compute the LQTS, for instance when dealing with quantum many-body systems 
(see Fig.~\ref{fig:LQTS_Ising} and the discussion in the Supplementary Information).
\vspace{0.5cm}

\begin{theorem} 
  For any assigned operators $X,Y$ satisfying the equation 
  \begin{equation}
    \label{eqW} 
    X W + W X = Y \,,
  \end{equation} 
 admit solutions of the form 
  \begin{eqnarray}
    W&=&W_0 + X^{\circleddash 1} Y R + R Y X^{\circleddash 1} \nonumber \\
     &&+\sum_{n=0}^{+\infty} (-1)^n X^n P Y (X^{\circleddash (n+1)}-R)\,,
  \end{eqnarray}
  where $X^{\circleddash 1}$ is the Moore Penrose pseudoinverse of $X$, 
  $R$ is the projector on kernel of $X$, $P=\mbox{\rm \openone} - R$ ({\rm $\openone$} indicates the identity matrix) and
  $W_0$ is an generic operator which anti-commutes with $Y$. 
  Furthermore if $X$ and $Y$ are Hermitian, equation~\eqref{eqW} admits solutions 
  which are Hermitian too: the latter can be expressed as 
  \begin{eqnarray}
    W&=&W_0 + X^{\circleddash 1} Y R + R Y X^{\circleddash 1} \\
    && +\frac{1}{2} \bigg[ \sum_{n=0}^{+\infty} (-1)^n X^n P Y (X^{\circleddash (n+1)}-R) + h.c.\bigg], \nonumber
  \end{eqnarray}
  where now $W_0$ is an arbitrary Hermitian operator which anti-commutes with $Y$. 
 \end{theorem} 

\begin{proof}
Since~\eqref{eqW} is a linear equation, a generic solution can be expressed as the sum 
of a particular solution plus a solution $W_0$ of the associated homogeneous equation, 
{\it i.e.} an operator which anti-commute with $X$, 
\begin{equation} 
  \label{eqW0}
  X W_0 + W_0 X = 0 \,.
\end{equation} 
A particular solution $W$ of equation~\eqref{eqW} can be always decomposed as
\begin{equation}
  W= RWR+RWP+PWR+PWP\,.
\end{equation}
Notice that by definition $RX=XR=O$, where $O$ identifies the null operator. 
Multiplying~\eqref{eqW} on both sides by $R$, one gets the condition $RYR=O$. 
The operator $W$, solution of equation~\eqref{eqW}, is defined up to its projection 
on the kernel subspace, that is
\begin{equation}
  W'= W + RWR \implies X W' + W' X = Y\,.
\end{equation}
Therefore, without loss of generality we can set
\begin{equation}
  RWR=O\,.
\end{equation}
Multiplying equation~\eqref{eqW} by $X^{\circleddash 1}$ on the right side 
and repeating the same operation on the left side we get:
\begin{eqnarray}
  R W P & = & R Y X^{\circleddash 1}\,, \\ 
  P W R & = & X^{\circleddash 1}YR \,.
\end{eqnarray}
On the other hand, $PWP$ satisfies the relation
\begin{equation}
  PWP= P Y X^{\circleddash 1} - P Y R - X (PWP) X^{\circleddash 1}\,.
\end{equation}
This equation can be solved recursively in $PWP$ and gives
\begin{equation}
  PWP= \sum_{n=0}^{\infty} (-1)^n X^n P Y (X^{\circleddash (1+n)} - R)\,,
\end{equation}
thus concluding the first part of the proof. 
The second part of the proof follows simply by observing that, if $X$ and $Y$ are Hermitian,
and if $W$ solves equation~\eqref{eqW}, then also its adjoint counterpart does. 
Therefore for each solution $W$ of the problem we can construct an Hermitian one 
by simply taking $(W+W^\dag)/2$. 
\end{proof}

\paragraph{ \textbf{Second-order term corrections to LQTS.}} 
In the low-temperature regime ($\beta\rightarrow \infty$), we have computed the second-order 
correction term to the LQTS, that is of $\mathfrak{S}_{\cal A}[\rho_\beta]$ in equation~\eqref{LQTS_def}, 
and found: 
\begin{eqnarray}
  \bigg( \!\! \frac{E_1 n_1}{n_0} \!\! \bigg)^{\!2} \! e^{- 2\beta E_1} \bigg\{ \!\! -\!2 \! + \! \mathrm{Tr}[P_0^{\cal A} \Pi_1^{\cal A} +\Pi_1^{\cal A}{\Pi_0^{\cal A}}^{\circleddash 1} \Pi_1^{\cal A}(1+P_0^{\cal A})] \nonumber \\ -\!2 \sum_{n=0}^{+\infty} (-1)^n \mathrm{Tr} \left[\Pi_1^{\cal A} {\Pi_0^{\cal A}}^{\circleddash (n+2)}\Pi_1^{\cal A}{\Pi_0^{\cal A}}^{n+1}\right] \! \bigg\} \;, \nonumber 
\end{eqnarray}
with $E_k \geq E_1$
and where the series in $n$ is meant to converge to $1/2$ when $\mathrm{Tr} \left[\Pi_1^{\cal A} {\Pi_0^{\cal A}}^{\circleddash (n+2)}\Pi_1^{\cal A}{\Pi_0^{\cal A}}^{n+1}\right]=1$, {\it i.e.}, 
$\sum_{n=0}^{+\infty} (-1)^n \equiv \lim_{x \to - 1} \sum_{n=0}^{+\infty} x^n =\frac{1}{2}$. 
In order to vanish, this second-order correction term requires a stronger condition 
with respect to one necessary to nullify the first-order term in the LQTS, equation~\eqref{EQQ}. 
It is given by $\Pi_1^{\cal A}=\Pi_0^{\cal A}$, and corresponds to the requirement that 
the system ground state must be locally indistinguishable from the first excited level.

\paragraph{ \textbf{Heat capacity in the 2-level Landau-Zener scheme.}} 

Here we discuss the simplified case in which only the ground state (with energy $E_0$) 
and the first excited level (with energy $E_1$) of the global system Hamiltonian $H$ play a role. 
In particular, we are interested in addressing a situation where 
the ground-state energy gap $\Delta E \equiv E_1 - E_0$ may become very small, as
a function of some external control parameter $\Gamma$ ({\it e.g.}, the magnetic field
or the system anisotropy). 
A sketch is depicted in Fig.~\ref{fig:LZ}, 
and refers to the so called Landau-Zener (LZ) model~\cite{LZ}. 
This resembles the usual scenario when a given many-body system 
is adiabatically driven, at zero temperature, across a quantum phase transition point. 

\begin{figure}[!t]
  \includegraphics[height=0.5\columnwidth]{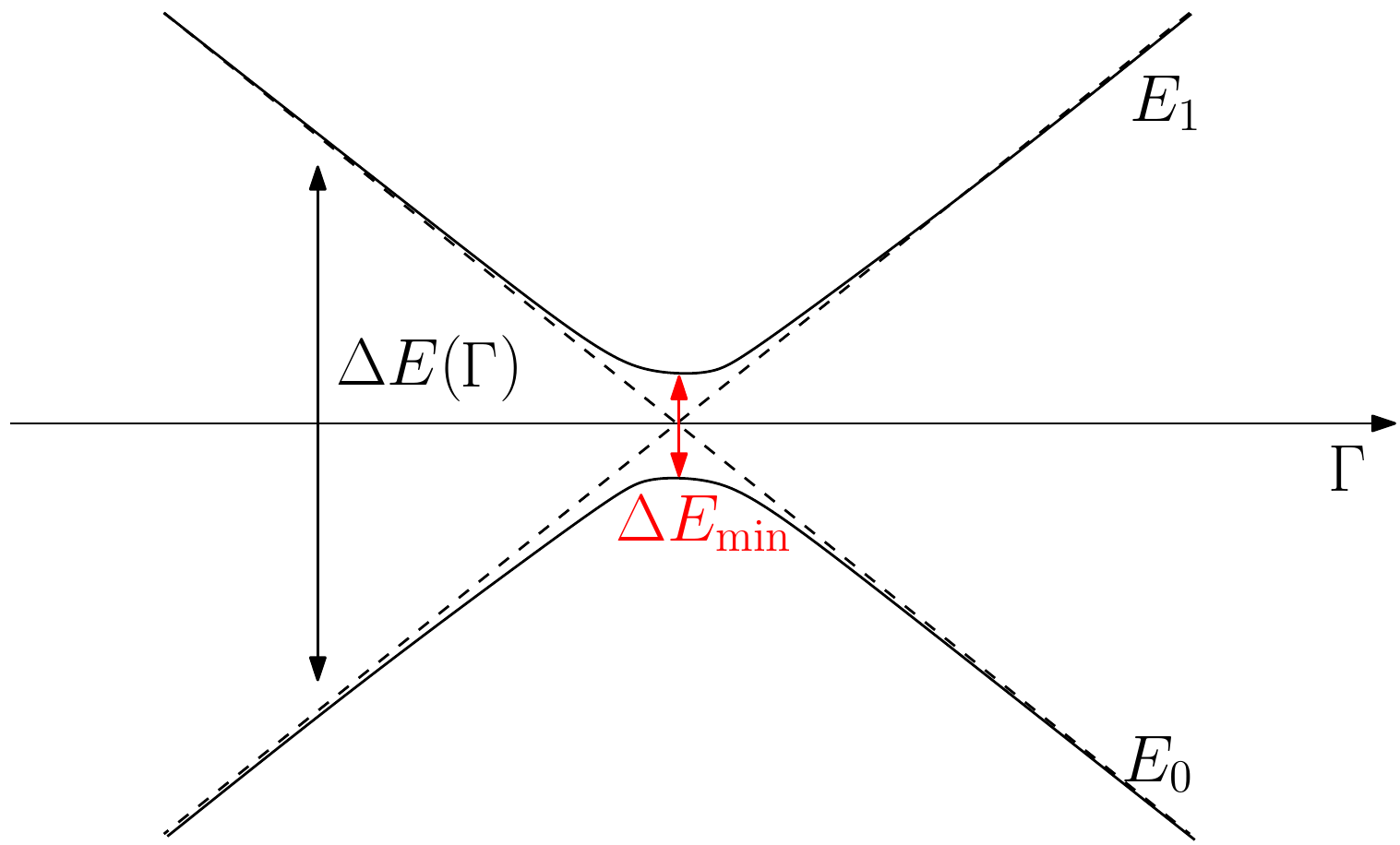}
  \caption{\textbf{The two-level Landau-Zener model.}
    A sketch of the behaviour of the two eigenenergies $(E_1,E_2)$ as a function 
    of some control parameter $\Gamma$.
    The gap $\Delta E = E_2-E_1$ displays a pronounced minimum in correspondence
    of a given $\Gamma^c$ value.}
  \label{fig:LZ}
\end{figure}

In correspondence of some critical value $\Gamma_c$, the gap is minimum. 
For a typical quantum many-body system, such minimum value $\Delta E_{\rm min}$
tends to close at the thermodynamic limit and a quantum phase transition occurs 
(notice that $\Gamma_c$ may depend on the system size).
Hereafter, without loss of generality, we will assume $E_0 = 0$ and take $E_1 = \Delta E$ 
so that the system heat capacity~\eqref{eq:Heat} reduces to:
\begin{equation}
  \mathfrak{S}_{\cal AB} [\rho_\beta] = \frac{n_1(\Delta E)^2 e^{-\beta \, \Delta E}}{n_0 + n_1 \, e^{-\beta \, \Delta E}} 
  - \left( \frac{n_1 \, \Delta E \, e^{-\beta \, \Delta E }}{n_0 + n_1 \, e^{-\beta \, \Delta E}} \right)^2 \, .
  \label{eq:LZthermsusc}
\end{equation}
Here $n_0$ and $n_1$ are the degeneracy indexes associated to the levels $E_0$ and $E_1$, respectively.
Notice that $\mathfrak{S}_{\cal AB}[\rho_\beta]$ is always non-negative and exhibits 
a non-monotonic behaviour as a function of $\Delta E$, at fixed $\beta$. 
Indeed it is immediate to see that $\mathfrak{S}_{\cal AB}[\rho_\beta] \to 0$
in both limits $\Delta E \to 0$ and $\Delta E \to +\infty$. 
For fixed $\beta, n_0$ and $n_1$, the heat capacity displays a maximum in correspondence 
of the solution of the transcendental equation 
\begin{equation}
  \frac{\partial \mathfrak{S}_{\cal AB}[\rho_\beta]}{\partial \Delta E}= 0 \iff e^{\beta \, \Delta E} = \frac{n_1}{n_0} \frac{(2 + \beta \, \Delta E)}{(\beta \, \Delta E - 2)}.
\end{equation}
In particular, for $n_0=n_1=1$ the latter relation is fulfilled for
$\Delta E^\star \approx 2.3994 / \beta$, while for $n_0=2, n_1=1$ 
it is fulfilled for $\Delta E^\star \approx 2.2278 / \beta$. 

It turns out that the behaviour of the heat capacity as a function of increasing $\Gamma$ 
in a two-level LZ scheme depends on the relative sizes of $\Delta E^\star$ and $\Delta E_{\rm min}$, 
as pictorially shown in Fig.~\ref{fig:LZ3}:
{\bf a)} if $\Delta E_{\rm min} > \Delta E^\star$, then $\mathfrak{S}_{\cal AB}[\rho_\beta]$
will exhibit a maximum in correspondence of $\Gamma_c$;
{\bf b)} if $\Delta E_{\rm min} < \Delta E^\star$, a maximum at $\Gamma_1^\star$ 
corresponding to $\Delta E = \Delta E^\star$ will appear, followed by a local minimum at $\Gamma_c$ 
and eventually by another maximum at $\Gamma_2^\star$ where the former condition occurs again.
Since $\Delta E^\star$ is a function of $\beta$, and $\Delta E_{\rm min}$ 
depends on the system size, the point of minimum gap can be signaled by a maximum 
or by a local minimum depending on the way the two limits $L \to +\infty$ (thermodynamic limit) 
and $\beta \to +\infty$ (zero-temperature limit) are performed. 

\begin{figure}[!t]
  \includegraphics[width=1\columnwidth]{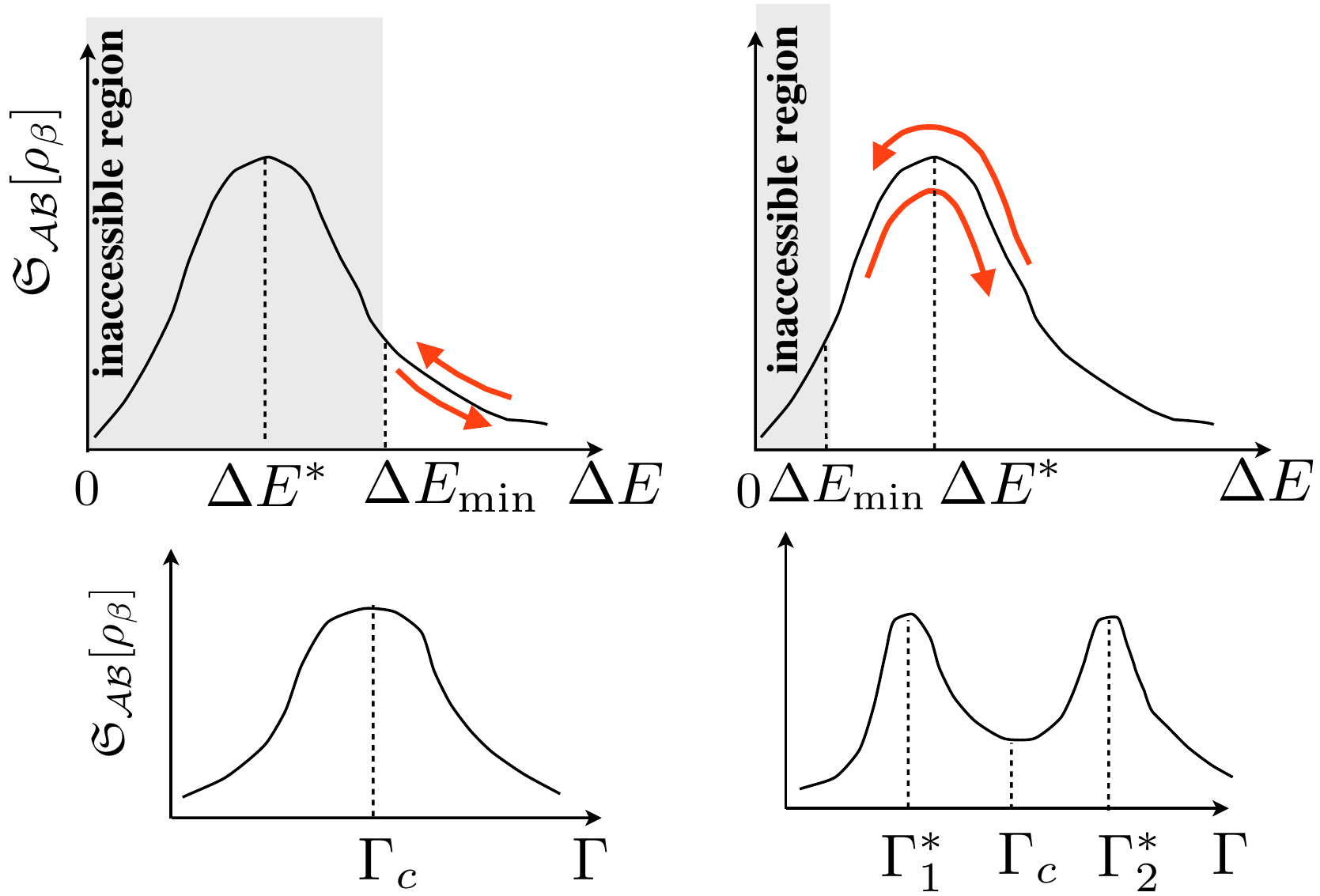}
  \caption{\textbf{Heat capacity in the LZ model.}
    The emergence of extremal points in the behaviour of $\mathfrak{S}_{\cal AB}[\rho_\beta]$ 
    as a function of increasing control parameter $\Gamma$ is associated to the relative 
    size of $\Delta E_{\rm min}$ with respect to $\Delta E^*$.
    The red arrows denote the changing of $\Delta E$ during a typical LZ protocol.
    One realizes that, if $\Delta E_{\rm min} > \Delta E^*$ one peak will appear (left),
    while if $\Delta E_{\rm min} < \Delta E^*$ one peaks will appear (right).}
  \label{fig:LZ3}
\end{figure}

\section{Acknowledgments}

We thank G. Benenti and G. De Chiara for useful discussions.
This work has been supported by MIUR through FIRB Projects RBFR12NLNA and PRIN ``Collective
quantum phenomena: from strongly correlated systems
to quantum simulators", by the EU Collaborative Project 
TherMiQ (Grant agreement 618074), and by the EU project COST Action MP1209 ``Thermodynamics in the 
quantum regime". 

\section{Author contributions}

All the authors conceived the work, agreed on the approach to pursue, analysed and discussed the results; 
A.D.P. performed the analytical calculations; 
D.R. performed the numerical calculations; V.G. and R.F. supervised the work.

\section{Additional material}

\paragraph{ \textbf{Supplementary Information}} accompanies this paper below.
\\

\paragraph{ \textbf{Competing financial interests:}} 
The authors declare no competing financial interests. 

\newpage 
\phantom{.}
\newpage


\section{Supplementary Information}

\subsection{Scaling of the LQTS with the dimension $n_{\cal A}$ \\ of subsystem ${\cal A}$.}

The local quantum thermal susceptibility (LQTS) for the subsystem ${\cal A}$
of a given system, as defined in equation~\eqref{LQTS_def}, is a quantity
which is tricky to be evaluated numerically. 
Apart from the exponential growth of the Hilbert space, extrapolating
the limit $\varepsilon \to 0$ typically requires high accuracies 
in the diagonalization procedure. This would generally limit the study of local 
thermometry in the many-body context up to very small systems.
The expression that we derived in equation~\eqref{eq:LQTS_numerical} circumvents
the latter problem and enables an easier manipulation of $\mathfrak{S}_{\cal A}[\rho_{\beta}]$.
However one still needs the full spectrum of the reduced density matrix $\rho_\beta^{\cal A}$, 
since all its eigenvectors have to be contracted with the product of the global equilibrium 
state times the system Hamiltonian, $\rho_\beta H$. This makes the whole analysis not 
straightforward, even for free-fermion systems as is the case for the Ising model $H_{\rm Ising}$.
For this reason we resort to an exact diagonalization technique.

\begin{figure}[!b]
  \includegraphics[width=0.8\columnwidth]{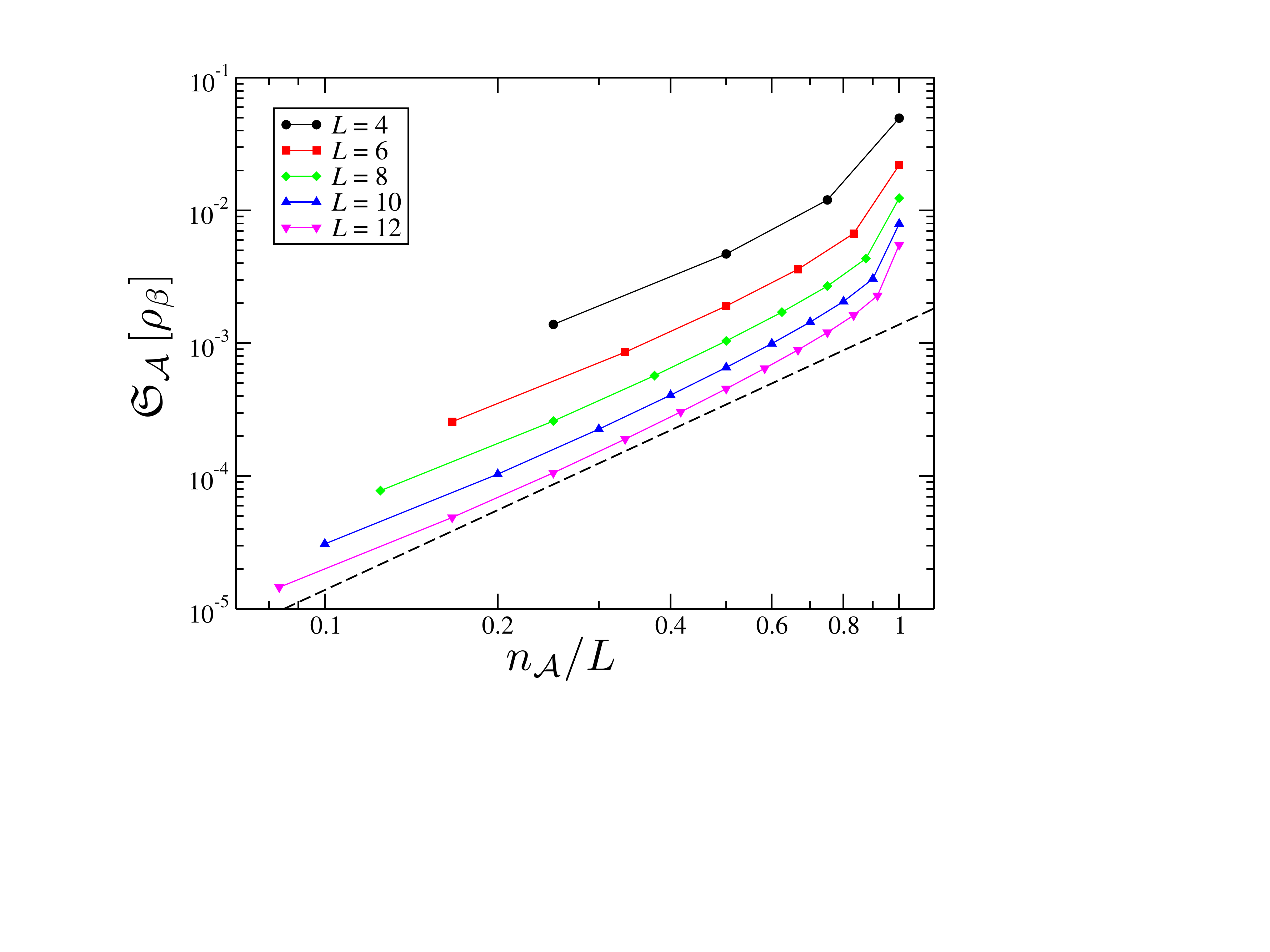}
  \caption{Peak value of $\mathfrak{S}_{\cal A}[\rho_{\beta}]$ for the Ising model, 
    as a function of the dimension of the measured subsystem $n_{\cal A}$. 
    Different data sets are for various system lengths $L$.
    The dashed line denotes a power-law behaviour $\mathfrak{S}_{\cal A} \sim (n/L)^2$, 
    and is plotted as a guideline. The temperature is chosen as a function of 
    the system size, as $\beta = 3L/4$.}
  \label{fig:IsingPeak_scal}
\end{figure}

Let us concentrate on the scaling of the LQTS 
with the size $n_{\cal A}$ of the measured subsystem for the Ising spin chain~\eqref{eq:Ising}. 
The maximum of $\mathfrak{S}_{\cal A}[\rho_{\beta}]$ as a function of $n_{\cal A}$,
and for different system sizes $L$ is shown in Fig.~\ref{fig:IsingPeak_scal}.
As noticed in the main text within the LZ discussion, the peak is related to
the appearance of the critical point, but deviations from its position
at the thermodynamic limit ($h_c=1$) are expected, due to the mutual interplay
of finite-size and finite-temperature effects. 
To balance them, we choose a size-dependent temperature $\beta = 3L/4$.
While with exact diagonalization we cannot go beyond the study of systems with size $L=12$, 
our data suggest a power-law behaviour of the peak value at least for very small subsystems,
scaling as
\begin{equation}
  \mathfrak{S}_{\cal A}[\rho_{\beta}] \sim (n/L)^\alpha,
  \label{eq:power-law}
\end{equation}
with an exponent $\alpha \approx 2$. 

We performed the same analysis also for the XXZ chain~\eqref{eq:XXZ},
and focused on the behaviour of $\mathfrak{S}_{\cal A}[\rho_{\beta}]$ around
the two critical points of the model, in correspondence to the two local minima.
As we did above, we fix a size-dependent temperature $\beta = 3L/4$
and study systems of size up to $L=12$.
In the left panel of Fig.~\ref{fig:XXZPeak_scal}, we concentrate
on the ferromagnetic point at $\Delta_{\rm f} = -1$.
Similarly to the critical point of the Ising model, a power-law scaling
of the type in equation~\eqref{eq:power-law} seems to emerge at small $n_{\cal A}$,
with an exponent $\alpha \approx 3$.
The scaling analysis at the antiferromagnetic point around $\Delta_{\rm af} = 1$ 
is less clear and probably requires larger system sizes (right panel). 
We notice the appearance of a cusp-like feature at $n_{\cal A} = L/2$, which reflects 
the different behaviour of $\mathfrak{S}_{\cal A}[\rho_{\beta}]$ for 
$\Delta \approx \Delta_{\rm af}$, depending on whether $n_{\cal A} \leq L/2$ 
or $n_{\cal A} > L/2$, as is clearly visible in the right panel of Fig.~\ref{fig:LQTS_Ising}. 
We have to stress that the departure of the two bunches of curves
becomes more evident at large values 
of $\beta$, while it tends to disappear when increasing the temperature.

\begin{figure}[!h]
  \includegraphics[width=0.95\columnwidth]{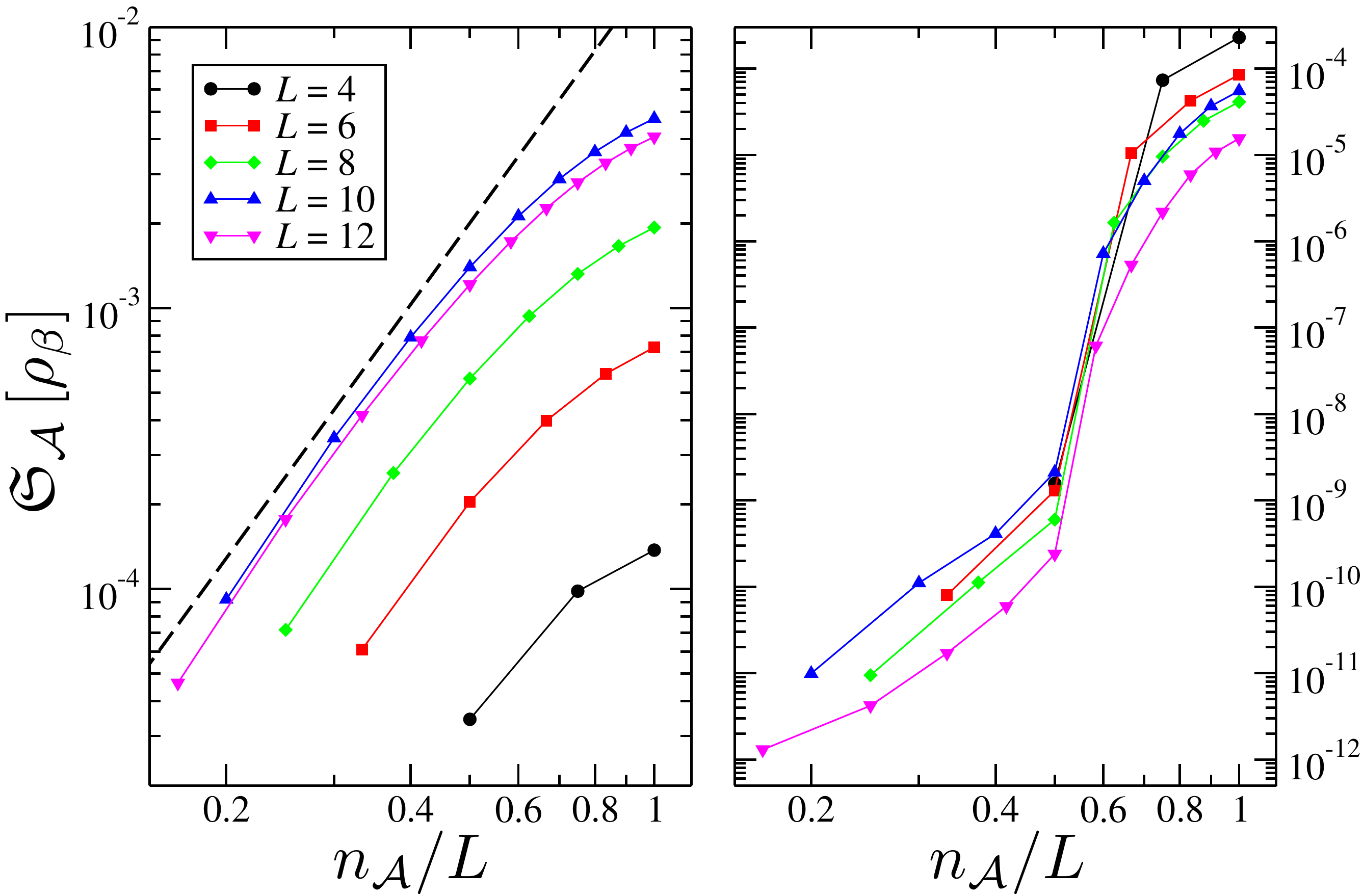}
  \caption{Same quantity and parameters as in Fig.~\ref{fig:IsingPeak_scal}, 
    but for the XXZ chain.
    The data in the two panels are for the minima of the LQTS close to 
    the critical points $\Delta_{\rm f} = -1$ (left) and $\Delta_{\rm af} = +1$ (right). 
    The dashed line in the left panel denotes the behaviour 
    $\mathfrak{S}_{\cal A} \sim (n/L)^3$ and is plotted as a guideline.}
  \label{fig:XXZPeak_scal}
\end{figure}

\subsection{Analysis of the heat capacity by considering \\ only the first excited levels.}

The heat capacity of the global system is quantified by the variance of the energy~\eqref{eq:Heat}, 
and it can be calculated easier than the LQTS, 
since it does not involve any partial tracing over a portion of the system
and depends only on the spectral properties.
In particular, at low temperatures $\beta \to +\infty$, the relevant contributions 
to $\mathfrak{S}_{\cal AB}[\rho_{\beta}]$ will be provided only by the first low-lying 
energy levels, as discussed within the Landau-Zener context.
Here we provide an analysis of the role of such excited states
in the two many-body systems that we address in this work.

We start from the Ising chain~\eqref{eq:Ising}.
The low-lying spectrum is shown in the upper panel of Fig.~\ref{fig:Ising_fewlev}, 
where we plot the gaps $\Delta E_j = E_j-E_0$ between the ground-state energy $E_0$ 
and those of the first excited states $E_j$, $j>0$.
Since we are considering the full Hilbert space of the system, 
for $L \to \infty$ the ground state in the ferromagnetic side ($|h|<h_c$) 
is doubly degenerate, while in the paramagnetic side ($|h|>h_c$) a gap opens up
monotonically as $\Delta E_1 = 2|h-h_c|$.
At finite values of $L$ (data with symbols), the double degeneracy survives
only at $h=0$ and $\Delta_1$ monotonically increases with $h>0$.

The scenario emerging by only keeping contributions to $\mathfrak{S}_{\cal AB}[\rho_{\beta}]$ 
coming from the ground state and the first excited state, 
while neglecting any other excited level, is quite clear.
In the thermodynamic limit and for finite $\beta$, the heat capacity is rigorously zero 
for $h<h_c$ and is finite for $h>h_c$, exhibiting a non-monotonic behaviour.
The position $h_c^{\rm *}$ of the maximum depends on the value of $\beta$ and
shifts toward $h_c=1$ as long as $\beta \to \infty$. Note however that 
$\mathfrak{S}_{\cal AB}[\rho_{\beta}] \stackrel{\beta \to \infty}{\longrightarrow} 0$.
In summary, the position of the maximum depends on a competition between 
the following two effects: $i)$ decreasing $L$ tends to shift the peak
toward the ferromagnetic phase ($h_c^{*} < h_c$); $ii)$ decreasing $\beta$ 
tends to shift the peak toward the paramagnetic phase ($h_c^{*} > h_c$).

\begin{figure}[!t]
  \includegraphics[height=0.72\columnwidth]{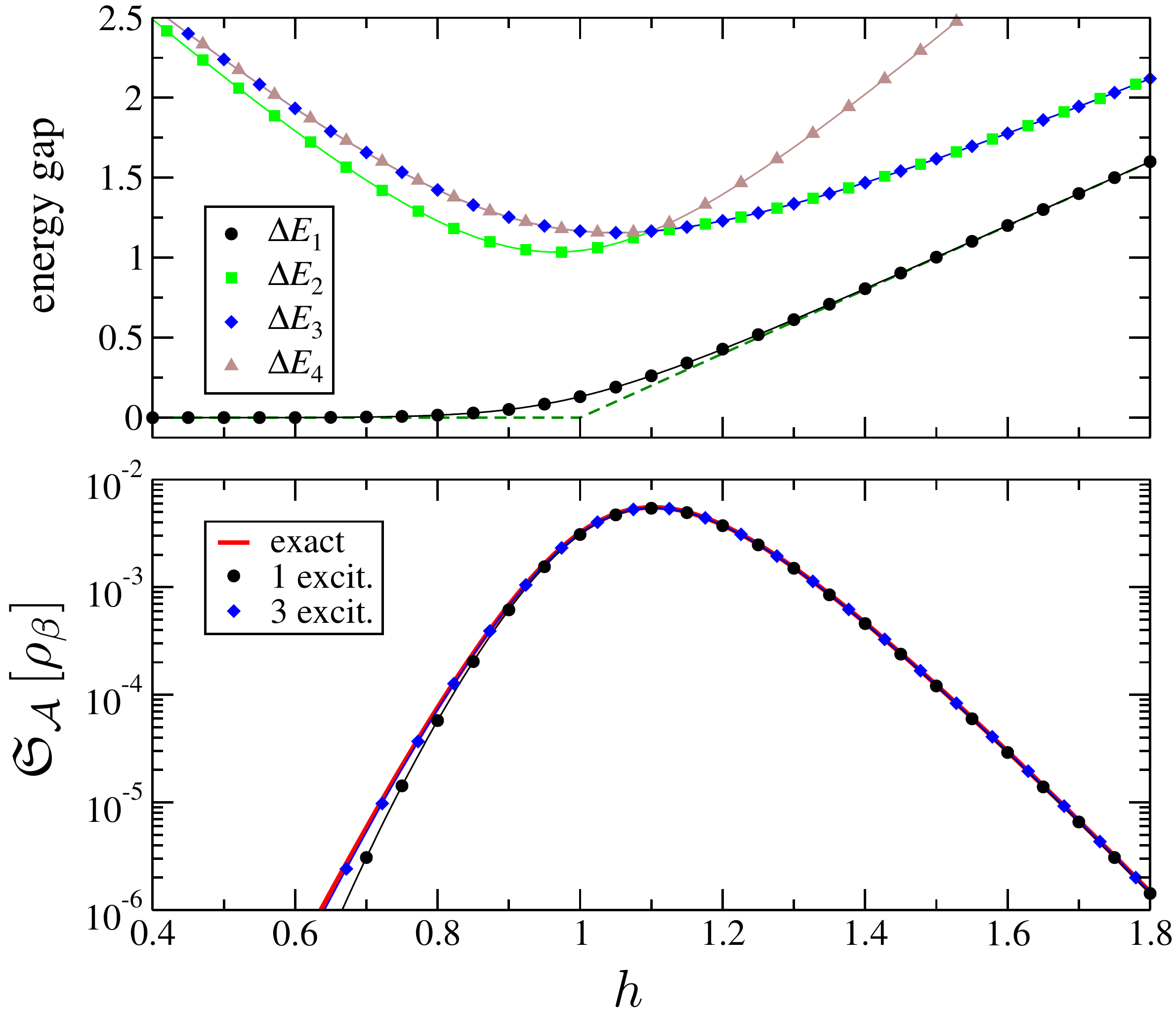}
  \caption{Heat capacity for the Ising chain,
    without breaking the ${\mathbb Z}_2$ symmetry.
    The upper panel shows the first energy gaps $\Delta E_j = E_j - E_0$
    as a function of the transverse field $h$. The dashed green line denotes
    the gap $\Delta E_1$ in the thermodynamic limit.
    The lower panel shows the global quantum thermal susceptibility 
    $\mathfrak{S}_{\cal AB}[\rho_{\beta}]$ for the full spectrum (red), 
    and when taking into account the first eigenenergies only.
    Here we consider a chain with $L=12$ and a temperature $\beta = 9$.}
  \label{fig:Ising_fewlev}
\end{figure}

In the bottom panel of Fig.~\ref{fig:Ising_fewlev} we observe that,
keeping only the ground and the first excited level in the computation 
of the energy variance~\eqref{eq:Heat}, already gives an excellent 
approximation to the exact value of $\mathfrak{S}_{\cal AB}[\rho_{\beta}]$.
Therefore the above situation effectively applies at large $\beta$ values.
In the presence of a symmetry breaking mechanism, one would
find a situation analogous to the LZ scheme: 
the presence of a maximum in proximity of $h_c$ would not be guaranteed, 
and a local minimum may appear (according to the values of $L$ and $\beta$).

Let us now consider the XXZ-Heisenberg chain~\eqref{eq:XXZ}.
Analogously as above, the upper panel of Fig.~\ref{fig:XXZ_fewlev} displays the first 
energy levels in the full Hilbert space (we do not break the symmetry associated 
to the conservation of the global magnetization along the $z$ axis).
For $\Delta < \Delta_{\rm f}$, the ground state is fully polarized along $z$,
and presents a double degeneracy at any length.
At $\Delta = \Delta_{\rm af}$ we see a cusp in the ground-state energy gap
that is due to a level crossing (this feature persists at larger sizes $L$).
This scenario characterizes the contribution to the heat capacity coming
only by the ground state and the first excited state (black line in the bottom panel).
In particular we see that, in this approximation, $\mathfrak{S}_{\cal AB}[\rho_{\beta}]$
is zero for $\Delta < \Delta_{\rm f}$, while it becomes finite at $\Delta > \Delta_{\rm f}$.
On the other side, the $\Delta_{\rm af}$ is signaled by a cusp, which displays a minimum,
and which raises as a consequence of the cusp in the ground-state energy gap.

Contrary to the Ising model, if also other excited levels are taken into account, 
non-negligible corrections to the energy variance appear
(see the bottom panel of Fig.~\ref{fig:XXZ_fewlev}).
In particular we note the emergence of a minimum around $\Delta_{\rm f}$, 
due to the fact that $\mathfrak{S}_{\cal AB}[\rho_{\beta}]$ becomes finite 
also for $J_z < \Delta_{\rm f}$.
The corner point at $\Delta_{\rm af}$ is smeared into a local minimum with a continuous derivative.
The positions of the two minima are influenced by finite-size and finite-temperature effects.

\begin{figure}[!b]
  \includegraphics[height=0.72\columnwidth]{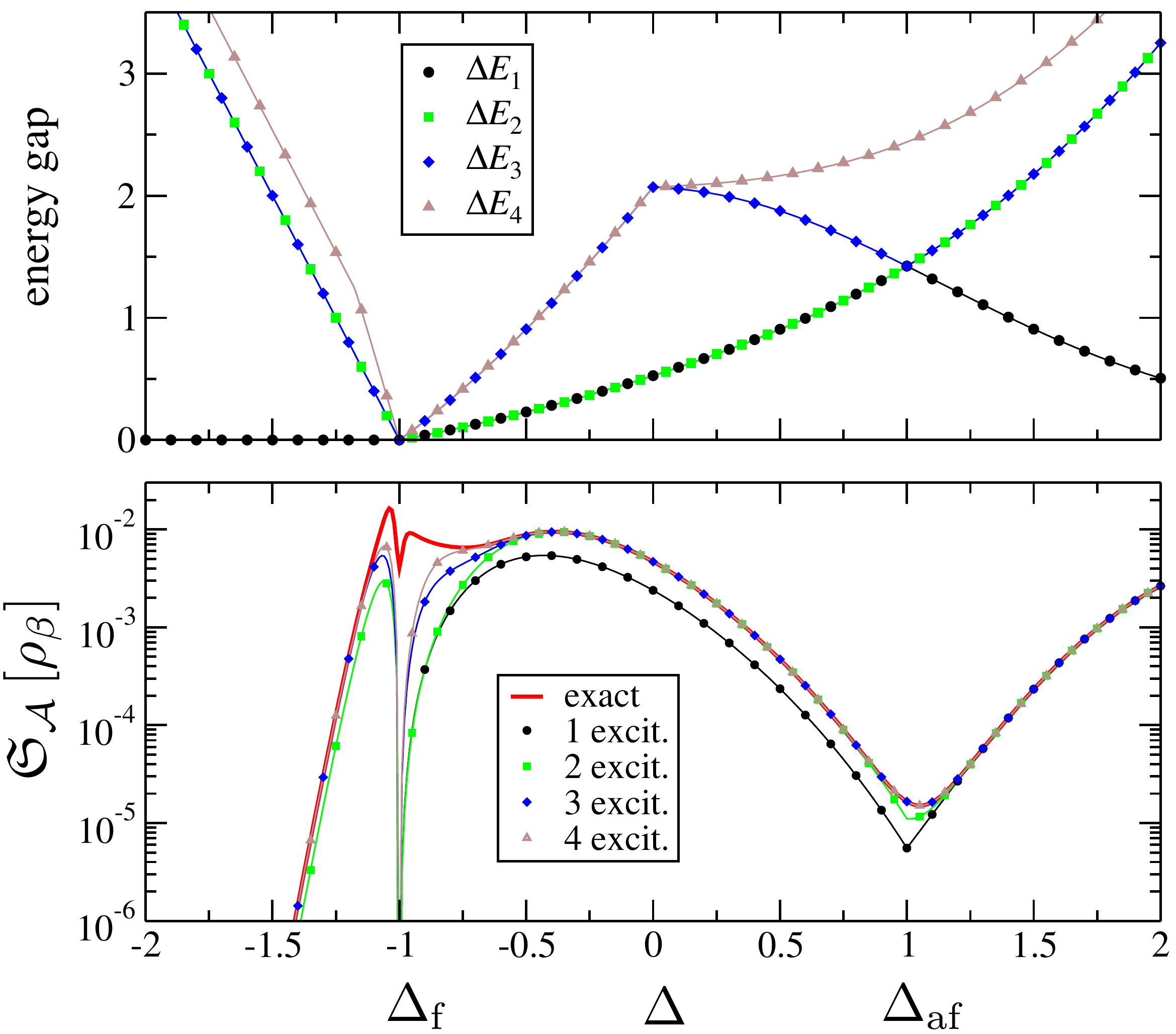}
  \caption{Same quantities and parameters as in Fig.~\ref{fig:Ising_fewlev}, 
    but for the XXZ chain in the full many-body Hilbert space of $L$ spins.}
  \label{fig:XXZ_fewlev}
\end{figure}

\end{document}